\documentclass{birkjour}
 \newtheorem{thm}{Theorem}[section]

 \newtheorem{prop}[thm]{Proposition}
 \theoremstyle{definition}
 \newtheorem{defn}[thm]{Definition}
 \theoremstyle{remark}

 \numberwithin{equation}{section}

\usepackage{graphicx}        
\usepackage{multicol}        
\usepackage[bottom]{footmisc}

\newcommand{\Pin}{\mathop{\mathrm{Pin}}}
\newcommand{\Spin}{\mathop{\mathrm{Spin}}}
\newcommand{\Cl}{\mathop{\mathrm{Cl}}}

\usepackage{lineno}
\usepackage{tikz}
\usepackage{pgfplots}
\usepgflibrary{shapes.geometric} 
\usepgflibrary[shapes.geometric] 
\usetikzlibrary{shapes.geometric} 
\usetikzlibrary[shapes.geometric] 

\begin{document}
	\bibliographystyle{plain} 
\small
%
%
%
%
%
%
%
%
%

\title[Group theory: Coxeter, conformal and modular groups]
 {Clifford algebra is the natural framework for root systems and Coxeter groups. Group theory: Coxeter, conformal and modular groups}

\author[Pierre-Philippe Dechant]{Pierre-Philippe Dechant}

\address{%
Department of Mathematics \\ University of York \\  Heslington, York, YO10 5GG \\ United Kingdom}

\email{ppd22@cantab.net}


 

\date{January 29, 2015}

\begin{abstract}
	In this paper, we make the case that Clifford algebra is the natural framework for root systems and reflection groups, as well as related groups such as the conformal and modular groups:  
	The metric that exists on these spaces can always be used to construct the corresponding Clifford algebra. 
	Via the Cartan-Dieudonn\'e theorem all the transformations of interest can be written as products of reflections and thus via `sandwiching' with  Clifford algebra multivectors. 
	These multivector groups can be used to perform concrete calculations in different groups, e.g. the various types of polyhedral groups, and we treat the example of the tetrahedral group $A_3$ in detail. 
	As an aside, this gives a constructive result that induces from every 3D root system a root system in dimension four, which hinges on the facts that the group of spinors provides a double cover of the rotations, the space of 3D spinors has a 4D euclidean inner product, and with respect to this inner product the group of spinors can be shown to be closed under reflections.
	In particular the 4D root systems/Coxeter groups induced in this way are precisely the exceptional ones, with the 3D spinorial point of view also explaining their unusual automorphism groups. This construction simplifies Arnold's trinities and puts the McKay correspondence into a wider framework. 
	We finally discuss extending the conformal geometric algebra approach to the 2D conformal and modular groups, which could have interesting novel applications in conformal field theory, string theory and modular form theory.
\end{abstract}

\maketitle
\section{Introduction}\label{sec_intro}

In previous work \cite{Dechant2012AGACSE, Dechant2012CoxGA, Dechant2012Induction, Dechant2013Platonic} we have commented on the usefulness of a Clifford algebra approach to reflection groups, and on the complementarity of these two  approaches. This was mostly due to the uniquely simple prescription for performing reflections in Clifford algebra, via a sandwiching procedure.  The Cartan-Dieudonn\'e theorem extends this to all orthogonal transformations by allowing to express them as products of reflections, and due to the homomorphism between $C(p,q)$ and $SO(p+1, q+1)$ this  in fact also extends to all conformal transformations. 

Here we argue that the Clifford framework is in fact the natural framework for root systems and reflection/Coxeter groups. In order to perform a reflection in the vector space in question, a metric is used to calculate the vector component parallel to the vector normal to the reflection hyperplane, which gets reversed under a reflection in that hyperplane. Thus, one can always use this metric on the vector space to construct the corresponding Clifford algebra. This setup then has the advantages of the approach above. The main route of interest in these structures is from a Lie theoretic perspective. A continuous Lie group, which is of interest for many local/gauge symmetry principles in high energy and gravitational physics as well as other areas, leads one to consider the Lie algebra of infinitesimal transformations. These Lie algebras usually have a triangular decomposition into the Cartan subalgebra (the quantum numbers) and the positive and negative roots (the creation and annihilation operators) of a crystallographic root system (which generates the Weyl group of the root lattice) on which the Killing form always provides a metric. One can thus always use this Killing form to construct the corresponding Clifford algebra. We therefore point out that it is more economical to employ this Clifford algebra structure (which is implicit anyway), rather than simply using the vector space and inner product structures separately. 

As an illustration of this method, we treat the example of the tetrahedral group associated with the root system $A_3$. This is particularly interesting, as it does not contain the inversion, which has led to a lot of confusion in the literature. Following the Cartan-Dieudonn\'e versor approach, we will demonstrate concrete group theoretic calculations for all the different types of polyhedral groups -- spinor, pinor, rotation and reflection -- as groups of multivectors (versors) in the Clifford algebra. 

Since the spinors in 3D have a 4D euclidean structure one gets as a by-product a remarkable result relating root systems in three and four dimensions: using a Clifford spinor construction one can construct -- starting from any root system in 3D -- a 4D root system, which generates a corresponding 4D reflection group. The 4D root systems that arise in this way are precisely the exceptional root systems in four dimensions, and their unusual symmetries can also be easily understood only from the 3D spinorial point of view. These results and the Clifford spinor construction do not appear to be known, and also present a more general context for the McKay correspondence. 

Via the homomorphism between $C(p,q)$ and $SO(p+1, q+1)$, all conformal transformations can in fact also be generated via reflections in a space of signature $(p+1, q+1)$, and the modular group, which is a subgroup of the 2D conformal group, can also be generated in this framework. We are  interested in conformal field theory and string theory, in particular in relation to string theory on the torus, modularity and Moonshine phenomena \cite{gannon2006moonshine, eguchi2011notes}, and therefore present a CGA construction of the 2D conformal and modular groups, which could have interesting novel applications.

 This paper is organised as follows. 
Section \ref{sec_Cox} introduces the framework of root systems, reflection and Coxeter groups, as well as their graphical representation. 
In Section \ref{sec_versor}, we present the Clifford algebra background, in particular a versor formalism for the Cartan-Dieudonn\'e theorem, by which the  full, chiral, binary and pinor polyhedral groups are all easily generated and treated within the same framework. 
In Section \ref{sec_spin}, we discuss the example of the root system of $A_3$, restricting ourselves to the even subgroups (spinors and rotations) for simplicity, and show how one can perform group theoretical calculations in this way; in particular, we discuss the conjugacy classes.
The spinors in 3D have a 4D euclidean norm, which yields a remarkable theorem which induces from every 3D root system a 4D root system in a constructive way (Section \ref{sec_bin}). We discuss this construction, the peculiar symmetries of the resulting root systems, and how this approach puts the McKay correspondence into a wider framework. 
In Section \ref{sec_pin}, we revisit the example of $A_3$ and treat the whole pinor and reflection groups in the versor formalism -- $A_3$ is the most salient example in this regard.
In Section \ref{sec_CGA}, we briefly outline how to extend the conformal geometric algebra treatment to the conformal and modular groups, which might have interesting new applications. 
We conclude with a summary and possible further work in Section \ref{sec_concl}.

\section{Root systems and reflection groups}\label{sec_Cox}

In this section, we introduce reflection groups via their generating root systems, which are sets of (root) vectors -- i.e. they can be interpreted as polytopes. Although perhaps unfamiliar, it is in fact a very simple concept, encapsulated in two straightforward axioms.

\begin{defn}[Root system] \label{DefRootSys}
A \emph{root system} is a collection $\Phi$ of non-zero  vectors $\alpha$ (called root vectors)  spanning an $n$-dimensional Euclidean vector space $V$ endowed with a positive definite bilinear form, which satisfies the  two axioms:
\begin{enumerate}
\item $\Phi$ only contains a root $\alpha$ and its negative, but no other scalar multiples: $\Phi \cap \mathbb{R}\alpha=\{-\alpha, \alpha\}\,\,\,\,\forall\,\, \alpha \in \Phi$. 
\item $\Phi$ is invariant under all reflections corresponding to root vectors in $\Phi$: $s_\alpha\Phi=\Phi \,\,\,\forall\,\, \alpha\in\Phi$, with 
the reflection $s_\alpha$ in the hyperplane that $\alpha$ is normal to given by $s_\alpha: \lambda\rightarrow s_\alpha(\lambda)=\lambda - 2\frac{(\lambda|\alpha)}{(\alpha|\alpha)}\alpha$, where $(\cdot \vert \cdot)$ denotes the inner product on $V$.
\end{enumerate}
\end{defn}

For a crystallographic root system, a subset $\Delta$ of $\Phi$, called \emph{simple roots} $\alpha_1, \dots, \alpha_n$, is sufficient to express every element of $\Phi$ via $\mathbb{Z}$-linear combinations with coefficients of the same sign. 
$\Phi$ is therefore  completely characterised by this basis of simple roots. In the case of the non-crystallographic root systems $H_2$, $H_3$ and $H_4$, the same holds for the extended integer ring $\mathbb{Z}[\tau]=\lbrace a+\tau b| a,b \in \mathbb{Z}\rbrace$, where $\tau$ is   the golden ratio $\tau=\frac{1}{2}(1+\sqrt{5})=2\cos{\frac{\pi}{5}}$.  
 For the crystallographic root systems, the classification in terms of Dynkin diagrams essentially follows the one familiar from Lie groups and Lie algebras, as their Weyl groups are precisely the crystallographic Coxeter groups. A mild generalisation to so-called Coxeter-Dynkin diagrams is necessary for the non-crystallographic root systems:
\begin{defn}[Coxeter-Dynkin diagram and Cartan matrix] 
	A graphical representation of the geometric content of a root system is given by \emph{Coxeter-Dynkin diagrams}, where nodes correspond to simple roots, orthogonal roots are not connected, roots at $\frac{\pi}{3}$ have a simple link, and other angles $\frac{\pi}{m}$ have a link with a label $m$. 
The  \emph{Cartan matrix} of a set of simple roots $\alpha_i\in\Delta$ is defined as the matrix
	$A_{ij}=2{(\alpha_i\vert \alpha_j)}/{(\alpha_j\vert \alpha_j)}$.

\end{defn}
For instance, the root system of the icosahedral group $H_3$ has one link labelled by $5$ (via the above relation $\tau=2\cos{\frac{\pi}{5}}$), as does its four-dimensional analogue $H_4$, and the infinite two-dimensional family $I_2(n)$ (the root systems of the symmetry groups of the regular $n$-gons) is labelled by $n$. A host of examples of diagrams is presented in Table \ref{tab:4} in Section \ref{sec_bin}.

The reflections in the second axiom of the root system generate a reflection group. A Coxeter group is a mathematical abstraction of the concept of a reflection as an involution (squares to the identity) as well as $m$-fold rotations (two successive reflections generate a rotation in the plane defined by the two roots)  in terms of abstract generators. 
\begin{defn}[Coxeter group] A \emph{Coxeter group} is a group generated by a set of involutive generators $s_i, s_j \in S$ subject to relations of the form $(s_is_j)^{m_{ij}}=1$ with $m_{ij}=m_{ji}\ge 2$ for $i\ne j$. 
\end{defn}
The  finite Coxeter groups have a geometric representation where the involutions are precisely realised as reflections at hyperplanes through the origin in a Euclidean vector space $V$, i.e. they are essentially  just the classical reflection groups. In particular, then the abstract generator $s_i$ corresponds to the simple {reflection}
$s_i: \lambda\rightarrow s_i(\lambda)=\lambda - 2\frac{(\lambda|\alpha_i)}{(\alpha_i|\alpha_i)}\alpha_i$
 at the hyperplane perpendicular to the  simple {root } $\alpha_i$.
The action of the Coxeter group is  to permute these root vectors, and its  structure is thus encoded in the collection  $\Phi\in V$ of all such roots, which in turn form a root system.

After explaining the reflection group framework, we now move onto the singularly simple prescription for performing reflections that Clifford algebra affords, in spaces of any dimension and signature.

\section{Versor framework}\label{sec_versor}

\begin{table}
\caption{Versor framework for a unified treatment of the chiral, full,  binary and pinor polyhedral groups, as well as the conformal and modular groups.}
\label{tab:1}       
%
%
\begin{tabular}{p{1.5cm}p{3.3cm}p{6.9cm}}
\hline\noalign{\smallskip}
Continuous group &Discrete subgroup & Multivector action  \\
\hline\noalign{\smallskip}
$SO(n)$&rotational/chiral & $x\rightarrow \tilde{R}xR$\\
$O(n)$&reflection/full & $x\rightarrow \pm\tilde{A}xA$\\
$\Spin(n)$&binary  & spinors $R$ under $(R_1,R_2)\rightarrow R_1R_2$\\
$\Pin(n)$& pinor & pinors $A$ under $(A_1,A_2)\rightarrow A_1A_2$\\
$C(p,q)$& conformal\&modular &  in a space of signature $(p+1, q+1)$ \\
\noalign{\smallskip}\hline\noalign{\smallskip}
\end{tabular}
\end{table}

The geometric product $xy=x\cdot y+x \wedge y$ (where the scalar product is the symmetric part and the exterior product the antisymmetric part) of Geometric/Clifford Algebra \cite{Hestenes1966STA, HestenesSobczyk1984, Hestenes1990NewFound,LasenbyDoran2003GeometricAlgebra} 
 provides a very compact and efficient way of handling reflections in any number of dimensions, and thus by the Cartan-Dieudonn\'e theorem in fact for any orthogonal transformation. For a unit vector $\alpha$, the two terms in the formula for a reflection of a vector $v$ in the hyperplane orthogonal to $\alpha$ from Definition (\ref{DefRootSys}) simplify to the double-sided action of $\alpha$ via the geometric product
	\begin{equation}\label{in2refl}
	  v\rightarrow s_\alpha v=v'=v-2(v|\alpha)\alpha=v-2\frac{1}{2}(v\alpha+\alpha v)\alpha=v-v\alpha^2-\alpha v\alpha=-\alpha v \alpha.
	\end{equation}
This  prescription for reflecting vectors in hyperplanes is remarkably compact; and by compounding reflections one is led to defining a versor as
 a multivector $A=a_1a_2\dots a_k$ which is the product of $k$ unit vectors $a_i$  \cite{Hestenes1990NewFound}.  Versors form a multiplicative group called the versor/pinor group $\Pin$ under the single-sided multiplication with the geometric product, with inverses given by $\tilde{A}A=A\tilde{A}=\pm 1$, where the tilde denotes the reversal of the order of the constituent vectors $\tilde{A}=a_k\dots a_2a_1$, and  the $\pm$-sign defines its parity.
Every orthogonal transformation $\underbar{A}$ of a vector $v$ can thus be expressed via unit versors via
\begin{equation}\label{in2versor}
\underbar{A}: v\rightarrow  v'=\underbar{A}(v)=\pm{\tilde{A}vA}.
\end{equation}
Unit versors are double-covers of the respective orthogonal transformation, as $A$ and $-A$ encode the same transformation. Even versors, called spinors or rotors, form a subgroup of the $\Pin$ group and a double covering of the special orthogonal group, called the $\Spin$ group.
Clifford algebra therefore affords a particularly natural and simple construction of the $\Spin$ groups. 
The versor realisation of the orthogonal group is much simpler than conventional matrix approaches. 
Table \ref{tab:1} summarises the various action mechanisms of multivectors: a rotation (e.g. the continuous group $SO(3)$ or the discrete subgroup the chiral tetrahedral group $T$) are given by double-sided action of a spinor $R$, whilst these spinors themselves form a group under single-sided action/multiplication (e.g. the continuous group $\Spin(3)\sim SU(2)$ or the discrete subgroup the binary tetrahedral group $2T$).
Likewise, a reflection (continuous $O(3)$ or the discrete subgroup the full tetrahedral group the Coxeter group $A_3$) corresponds to sandwiching with the versor $A$, whilst the versors single-sidedly form a multiplicative group (the $\Pin(3)$ group or the discrete analogue, the double cover of $A_3$, which we denote $\Pin(A_3)$). As we shall see later in the CGA setup one uses the fact that the conformal group   $C(p,q)$ is homomorphic to $SO(p+1,q+1)$ to treat translations as well as rotations in a unified versor framework (see Section \ref{sec_CGA}). Thus conformal and modular transformations are carried out by sandwiching with multivectors in a representation space with signature $(p+1,q+1)$.

\begin{equation}\label{in2PA}
  \underbrace{\{1\}}_{\text{1 scalar}} \,\,\ \,\,\,\underbrace{\{e_1, e_2, e_3\}}_{\text{3 vectors}} \,\,\, \,\,\, \underbrace{\{e_1e_2=Ie_3, e_2e_3=Ie_1, e_3e_1=Ie_2\}}_{\text{3 bivectors}} \,\,\, \,\,\, \underbrace{\{I\equiv e_1e_2e_3\}}_{\text{1 trivector}}.
\end{equation}
The Geometric Algebra of three dimensions $\Cl(3)$ is spanned by three orthogonal (thus anticommuting) unit vectors $e_1$, $e_2$ and $e_3$, and contains the three bivectors $e_1e_2$, $e_2e_3$ and $e_3e_1$ that all square to $-1$, as well as the  highest grade object $e_1e_2e_3$   (trivector and pseudoscalar), which also squares to $-1$.

\section{$A_3$: the rotation and spinor groups }\label{sec_spin}

Here we work out a detailed case, for which we choose the tetrahedral group $A_3$ as both the simplest and most salient example. $A_3$ is the smallest out of the polyhedral groups (the octahedral group $B_3$ is twice the order and the icosahedral group $H_3$ is five times the order). However, it is also the only one that does not contain the inversion -- this means that the $\Pin$ double cover does not simply consist of two copies of the $\Spin$ group, as is the case for $B_3$ and $H_3$. This point has led to much confusion in the literature, in particular as regards pure quaternions \cite{MoodyPatera:1993b, Humphreys1990Coxeter, Koca2006F4}, which was cleared up in the context of the induction theorem of the next section in \cite{Dechant2012AGACSE, Dechant2012CoxGA, Dechant2012Induction, Dechant2013Platonic}. 

We choose as our set of simple roots for  $A_3$ $\alpha_1=\frac{1}{\sqrt{2}}(e_2-e_1)$, $\alpha_2=\frac{1}{\sqrt{2}}(e_3-e_2)$ and $\alpha_3=\frac{1}{\sqrt{2}}(e_1+e_2)$. Taking the closure of these simple roots under the simple reflections generated by them generates a root system consisting of 12 roots, which form the vertices of a cuboctahedron. One can then multiply these root vectors together in the Clifford algebra to form higher grade multivectors that encode various different geometric transformations. However, rather than freely multiplying these vectors together, in this section we focus on even products of root vectors, e.g. $\alpha_1\alpha_2=\frac{1}{2}(-1+e_1e_2+e_2e_3+e_3e_1)$ etc. This yields a set of 24 spinors, $\Spin(A_3)$, which is the binary polyhedral group $2T$. It is the even subgroup of the pinor group of order 48 and gives a double cover of the rotational/chiral tetrahedral group of order 12. 

\begin{table}
\caption{Conjugacy classes of the binary tetrahedral group $2T$ with group elements as spinors in the versor framework.}
\label{tab:2}       
%
%
\begin{tabular}{p{2.7cm}p{8.9cm}}
\hline\noalign{\smallskip}
Conjugacy Class & Group elements  \\
\hline\noalign{\smallskip}
$\mathbf{1}$&$1$\\
\hline\noalign{\smallskip}
$\mathbf{1}_-$&$-1$\\
\hline\noalign{\smallskip}
$\mathbf{4}$&
$\frac{1}{2}\left(1-e_1e_2+e_2e_3-e_3e_1\right)$, 
$\frac{1}{2}\left(1-e_1e_2-e_2e_3+e_3e_1\right)$, \\
&$\frac{1}{2}\left(1+e_1e_2-e_2e_3-e_3e_1\right)$, 
$\frac{1}{2}\left(1+e_1e_2+e_2e_3+e_3e_1\right)$\\
\hline\noalign{\smallskip}
$\mathbf{4_-}$&
$-\frac{1}{2}\left(1-e_1e_2+e_2e_3-e_3e_1\right)$, 
$-\frac{1}{2}\left(1-e_1e_2-e_2e_3+e_3e_1\right)$, \\
&$-\frac{1}{2}\left(1+e_1e_2-e_2e_3-e_3e_1\right)$, 
$-\frac{1}{2}\left(1+e_1e_2+e_2e_3+e_3e_1\right)$\\
\hline\noalign{\smallskip}
$\mathbf{4}^{-1}$&
$\frac{1}{2}\left(1+e_1e_2-e_2e_3+e_3e_1\right)$, 
$\frac{1}{2}\left(1+e_1e_2+e_2e_3-e_3e_1\right)$, \\
&$\frac{1}{2}\left(1-e_1e_2+e_2e_3+e_3e_1\right)$, 
$\frac{1}{2}\left(1-e_1e_2-e_2e_3-e_3e_1\right)$\\
\hline\noalign{\smallskip}
$\mathbf{4}^{-1}_-$&
$-\frac{1}{2}\left(1+e_1e_2-e_2e_3+e_3e_1\right)$, 
$-\frac{1}{2}\left(1+e_1e_2+e_2e_3-e_3e_1\right)$, \\
&$-\frac{1}{2}\left(1-e_1e_2+e_2e_3+e_3e_1\right)$,
$-\frac{1}{2}\left(1-e_1e_2-e_2e_3-e_3e_1\right)$\\
\hline\noalign{\smallskip}
$\mathbf{6}$&$\pm e_1e_2$, $\pm e_2e_3$, $\pm e_3e_1$\\
\hline\noalign{\smallskip}
\end{tabular}
\end{table}

With these sets of multivectors one can perform standard group theory calculations, such as finding the inverse $R^{-1}$, and then finding conjugacy classes of elements $R_1$ and $R_2$ related via conjugation $R_2=\tilde{R}R_1{R}$. Table \ref{tab:2} details the set of conjugacy classes for the spinor group, the binary tetrahedral group $2T$. It has 24 elements and 7 conjugacy classes, and calculations in the algebra are very straightforward; standard approaches would have to somehow move from $SO(3)$ rotation matrices to $SU(2)$ matrices for the binary group -- here both are treated in the same framework. 

The elements $1$ and $-1$ are of course in conjugacy classes by themselves ($\mathbf{1}$ and $\mathbf{1_-}$), and then there are four conjugacy classes of order 4. As spinors, these are precisely one class which we denote by $\mathbf{4}$ (e.g. contains $R$), its negative $\mathbf{4}_-$ (contains $-R$), its inverse $\mathbf{4}^{-1}$ (contains $\tilde{R}$) and the negative of its inverse $\mathbf{4}^{-1}_-$ (contains $-\tilde{R}$). The remaining class $\mathbf{6}$ of order 6 consists of the 6 permutations of the unit bivectors, up to sign. This gives the orders of $1$, $1$, $4$, $4$, $4$, $4$, $6$, as expected for $2T$. Since these spinors doubly cover the rotational tetrahedral group $T$ of order $12$, the conjugacy classes of $2T$ collapse down to those of $T$ as follows: $1$ and $-1$ are the same under double-sided action and give the identity in $T$. The conjugacy class of order 4 and its negative give the same rotation, so result in a $\mathbf{4}$ in $T$; likewise its inverse and the negative of the inverse give another $\mathbf{4}^{-1}$ in $T$. The six bivectors in $2T$ when taken up to sign result in a halving of the class, i.e. a conjugacy class of three in $T$. This gives the expected four conjugacy classes of orders $\mathbf{1}$, $\mathbf{3}$, $\mathbf{4}$ and $\mathbf{4}^{-1}$ for the chiral tetrahedral group $T$, as shown in Table \ref{tab:3}.

\begin{table}
\caption{Conjugacy classes of the chiral tetrahedral group $T$ as performed via double-sided action of the spinors (hence $\pm$-signs encode the same rotation).}
\label{tab:3}       
%
%
\begin{tabular}{p{2.7cm}p{8.9cm}}
\hline\noalign{\smallskip}
Conjugacy Class & Distinct rotations given by two spinors each ($\pm$)\\
\hline\noalign{\smallskip}
$\mathbf{1}$&$\pm 1$\\
\hline\noalign{\smallskip}
$\mathbf{4}$&
$\pm\frac{1}{2}\left(1-e_1e_2+e_2e_3-e_3e_1\right)$, 
$\pm\frac{1}{2}\left(1-e_1e_2-e_2e_3+e_3e_1\right)$, \\
&$\pm\frac{1}{2}\left(1+e_1e_2-e_2e_3-e_3e_1\right)$, 
$\pm\frac{1}{2}\left(1+e_1e_2+e_2e_3+e_3e_1\right)$\\
\hline\noalign{\smallskip}
$\mathbf{4}^{-1}$&
$\pm\frac{1}{2}\left(1+e_1e_2-e_2e_3+e_3e_1\right)$, 
$\pm\frac{1}{2}\left(1+e_1e_2+e_2e_3-e_3e_1\right)$, \\
&$\pm\frac{1}{2}\left(1-e_1e_2+e_2e_3+e_3e_1\right)$, 
$\pm\frac{1}{2}\left(1-e_1e_2-e_2e_3-e_3e_1\right)$\\
\hline\noalign{\smallskip}
$\mathbf{3}$&$\pm e_1e_2$, $\pm e_2e_3$, $\pm e_3e_1$\\
\hline\noalign{\smallskip}
\end{tabular}
\end{table}

Defining the exponential of a unit bivector $B$ in the usual way via a power series expansion results in the spinor $\exp(B\theta)=\cos \theta + B \sin \theta$, in analogy to de Moivre's formula, since $B^2=-1$ (this actually gives a rotation by $2\theta$). We can thus try to reexpress the spinors above as bivector exponentials, where the bivectors give the rotation planes (since two vectors define a plane) and play the role of the usual unit imaginary $i$ in these. This thus gives $n$-fold rotations via a spinor $\exp(\pi B/n)$.

One finds that 16 group elements can be written as $\pm \exp(\pi B/3)$ for the 8 permutations of $\sqrt{3}B=\pm e_1e_2\pm e_2e_3\pm e_3e_1$. These spinors are precisely those giving  3-fold rotations, i.e. 
the bivectors give the four planes of 3-fold symmetry of the cuboctahedron and tetrahedron  (vertices and faces). The 4 types of spinors $\pm \exp(\pm\pi B/3)$ for four different bivectors therefore correspond to the four $\mathbf{4}$ conjugacy classes above, since the two $\pm$ signs encode $-R$ and $\tilde{R}$, respectively. 
The six bivectors $\pm e_1e_2 = \pm\exp(e_1e_2 \pi/2)=\pm\exp(-e_1e_2 \pi/2)$, $\pm e_2e_3$, $\pm e_3e_1$ generate 2-fold rotations (by $\pi$, irrespective of whether $\theta$ is $\pm \pi/2$), which are the three (orthogonal) rotation planes with 2-fold symmetry of the tetrahedron (edges). 
This accounts for 16 and 6 elements, and including the identity and minus the identity one has 24 in total, as expected.

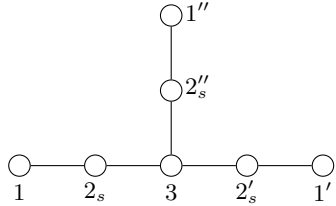
\begin{figure}
	\begin{center}
\begin{tikzpicture}[scale=0.5,
knoten/.style={        circle,      inner sep=.1cm,        draw}
]
\node at (-1,.7) (knoten0) [knoten,  color=white!0!black] {};
\node at  (1,.7) (knoten1) [knoten,  color=white!0!black] {};
\node at  (3,.7) (knoten2) [knoten,  color=white!0!black] {};
\node at  (5,.7) (knoten3) [knoten,  color=white!0!black] {};
\node at  (7,.7) (knoten4) [knoten,  color=white!0!black] {};

\node at  (3,2.7) (knoten5) [knoten,  color=white!0!black] {};
\node at (3,4.7) (knoten6) [knoten,  color=white!0!black] {};

\node at  (-1,0) (alpha0)  {$1$};
\node at  (1,0)  (alpha1) {$2_s$};
\node at  (3,0)  (alpha2) {$3$};
\node at  (5,0)  (alpha3) {$2_s'$};
\node at  (7,0)  (alpha4) {$1'$};

\node at (3.7,2.8)  (alpha5) {$2_s''$};
\node at (3.7,4.8) (alpha8) {$1''$};

\path  (knoten0) edge (knoten1);
\path  (knoten1) edge (knoten2);
\path  (knoten2) edge (knoten3);
\path  (knoten3) edge (knoten4);
\path  (knoten2) edge (knoten5);
\path  (knoten5) edge (knoten6);

\end{tikzpicture} 
\end{center}
\caption[$E_6^+$]{The graph depicting the tensor product structure of the binary tetrahedral group $2T$ is the same as the Dynkin diagram for the  affine extension of $E_6$, $E_6^+$. }
\label{figE6aff}
\end{figure}

The four conjugacy classes of $T$ of order $12$ determine that this group has four irreducible representations of dimensions  $1$, $1'$, ${1}''$ and $3$ (since the sum of the dimensions squared gives the order of the group $\sum d_i^2=|G|$), and the seven conjugacy classes of $2T$ of order $24$ mean this acquires a further three irreducible spinorial representations of dimensions $2_s$, $2_s'$, ${2_s}''$. The binary tetrahedral group has a curious connection with the affine Lie algebra $E_6^+$ (and likewise for the other binary polyhedral groups and the affine Lie algebras of $ADE$-type) via the so-called McKay correspondence, which is twofold: 
 We can define a graph by assigning a node to each irreducible representation of the binary tetrahedral group with the following rule for connecting nodes: each node corresponding to a certain irreducible representation is connected to the nodes corresponding to those irreducible representations that are contained in its tensor product with the irrep $2_s$. For instance, tensoring the trivial representation $1$ with $2_s$ trivially gives $2_s$ and thus the only link  $1$ has is with $2_s$; $2_s\otimes 2_s=1+3$, such that $2_s$ is connected to $1$ and $3$, etc. The graph that is built up in this way is precisely the Dynkin diagram of affine $E_6$, as shown in Figure \ref{figE6aff}. The second connection is the following observation: the Coxeter element is the product of all the simple reflections $\alpha_1\dots \alpha_6$ and its order, the Coxeter number $h$, is $12$ for $E_6$. This also happens to be the sum of the dimensions of the irreducible representations of $2T$, $\sum d_i$. This extends to all other cases of polyhedral groups and $ADE$-type affine Lie algebras.

\section{A 3D spinorial view of 4D exceptional phenomena}\label{sec_bin}
On this curious note, we make a profound observation: the 3D spinors generated by a 3D root system can be used to construct from any given 3D root system a root system in 4 dimensions in the following way.
\begin{prop}[$O(4)$-structure of spinors]\label{HGA_O4}
The space of $\Cl(3)$-spinors $R=a_0+a_1Ie_1+a_2Ie_2+a_3Ie_3$ can be endowed with a \emph{4D Euclidean norm} $|R|^2=R\tilde{R}=a_0^2+a_1^2+a_2^2+a_3^2$ induced by the  \emph{inner product} $(R_1,R_2)=\frac{1}{2}(R_1\tilde{R}_2+R_2\tilde{R}_1)$ between  two spinors $R_1$ and $R_2$. 
\end{prop}
This allows one to check that the group of 3D spinors generated from a 3D root system can be interpreted as a set of 4D vectors, which in fact surprisingly satisfies the axioms of a root system as given in Definition \ref{DefRootSys}. 
\begin{thm}[Induction Theorem]\label{HGA_4Drootsys}
Any 3D root system gives rise to a spinor group $G$ which induces a root system in 4D.
\end{thm}
\begin{proof}
Check the two axioms for the root system $\Phi$ consisting of the set of 4D vectors given by the 3D spinor group through Proposition \ref{HGA_O4}.
\begin{enumerate}
\item By construction, $\Phi$ contains the negative of a root $R$ since spinors provide a double cover of rotations, i.e. if $R$ is in a spinor group $G$, then so is $-R$ , but no other scalar multiples (normalisation to unity). 
\item $\Phi$ is invariant under all reflections with respect to the inner product $(R_1,R_2)$ in Proposition \ref{HGA_O4} since $R_2'=R_2-2(R_1, R_2)/(R_1, {R}_1) R_1=-R_1\tilde{R}_2R_1\in G$ for $R_1, R_2 \in G$ by the closure property of the group $G$ (in particular $-R$ and $\tilde{R}$ are in $G$ if $R$ is). 
\end{enumerate}
\end{proof}
Since the number of irreducible 3D root systems is limited to the trinity $(A_3, B_3, H_3)$ \cite{Arnold2000AMS,Arnold1999symplectization}, this yields a definite list of induced root systems in 4D -- this turns out to be the trinity $(D_4, F_4, H_4)$, which are exactly the exceptional root systems in 4D. We note in passing that Arnold's original link between these two trinities is extremely convoluted, and that our contruction presents a novel and direct link between the two.

\begin{table}
\caption{Clifford spinor construction and McKay correspondence: connections between 3D, 4D and $ADE$-type root systems. $|\Phi|$, $\sum d_i$ and $h$ are $6$, $12$, $18$ and $30$, respectively.}
\label{tab:4}       
%
%
\begin{tabular}{p{0.4cm}p{3cm}p{0.4cm}p{3cm}p{6.9cm}}
\hline\noalign{\smallskip}
&3D group &&4D group & affine Lie algebra  \\
\hline\noalign{\smallskip}
$A_1^3$ &\includegraphics[width=2cm]{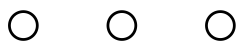}&$A_1^4$ &\includegraphics[width=2cm]{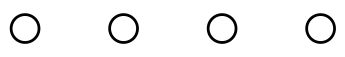}& $D_4^+$ \includegraphics[width=2cm]{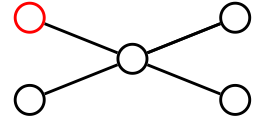}\\
\hline\noalign{\smallskip}
$A_3$ &\includegraphics[width=2cm]{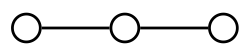}&$D_4$ &\includegraphics[width=2cm]{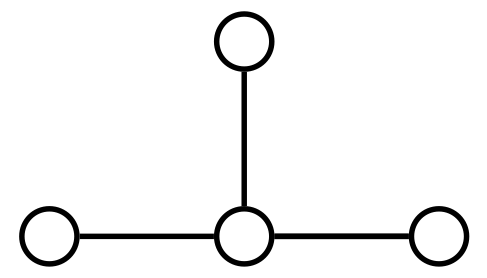}& $E_6^+$ \includegraphics[width=2cm]{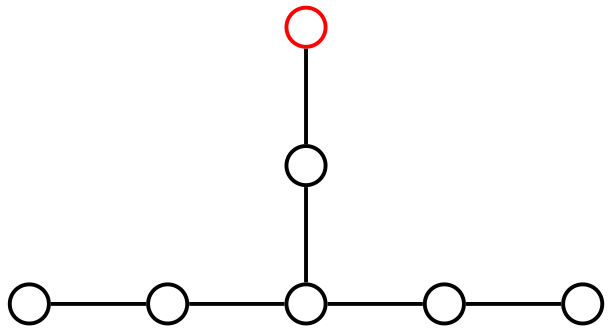}\\
\hline\noalign{\smallskip}
$B_3$ &\includegraphics[width=2cm]{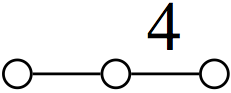}&$F_4$ &\includegraphics[width=2cm]{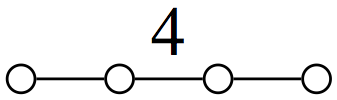} & $E_7^+$ \includegraphics[width=3cm]{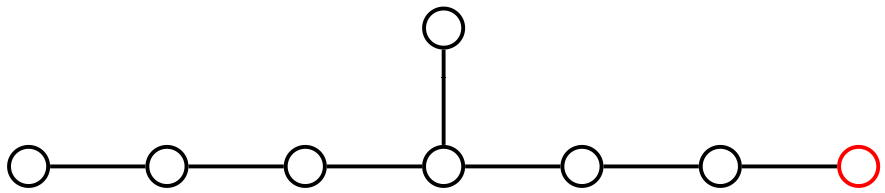}\\
\hline\noalign{\smallskip}
$H_3$ &\includegraphics[width=2cm]{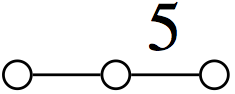}&$H_4$ &\includegraphics[width=2cm]{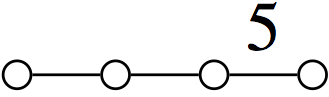} & $E_8^+$ \includegraphics[width=4cm]{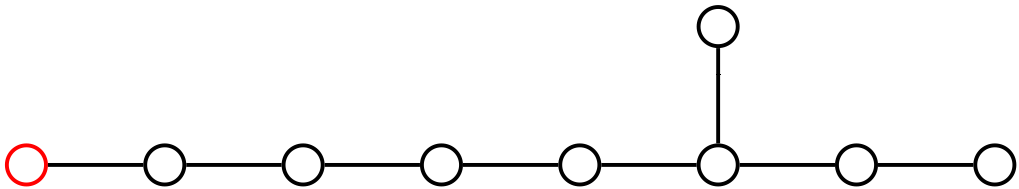}\\
\hline\noalign{\smallskip}
\end{tabular}
\end{table}

Not only is the existence of these root systems exceptional (i.e. they don't have counterparts in other dimensions), but their automorphism groups are also highly unusual. The 3D spinorial interpretation also explains these straightforwardly by the following theorem.
\begin{thm}[Spinorial symmetries]\label{HGAsymmetry}
A root system induced via the Clifford spinor construction has an automorphism group that contains two factors of the respective spinor group $G$ acting from the left and from the right.
\end{thm}
We noted earlier that the spinor groups and the $ADE$-type affine Lie algebras are connected via the McKay correspondence, for instance the binary polyhedral groups $(2T, 2O, 2I)$ and the Lie algebras $(E_6, E_7, E_8)$ -- for these $(12, 18, 30)$ is both the  Coxeter number of the respective Lie algebra and the sum of the dimensions of the irreducible representation of the polyhedral group.
However, the connection between $(A_3, B_3, H_3)$ and $(E_6, E_7, E_8)$ does not seem to be known. In particular, we note that $(12, 18, 30)$ is exactly the number of roots $\Phi$ in the 3D root systems $(A_3, B_3, H_3)$. Our construction therefore makes deep connections between trinities, and puts the McKay correspondence in a wider framework, as shown in Table \ref{tab:4}.

\section{$A_3$: the reflection and pinor groups }\label{sec_pin}

%
\begin{table}
\caption{Conjugacy classes of the pinor tetrahedral group with group elements as pinors in the versor framework.}
\label{tab:5}       
%
%
\begin{tabular}{p{2.8cm}p{8.9cm}}
\hline\noalign{\smallskip}
Conjugacy Class & Group elements  \\
\hline\noalign{\smallskip}
$\mathbf{1}$&$1$\\
\hline\noalign{\smallskip}
$\mathbf{1}_-$&$-1$\\
\hline\noalign{\smallskip}
$\mathbf{8}_+$&
$\frac{1}{2}\left(1-e_1e_2+e_2e_3-e_3e_1\right)$, 
$\frac{1}{2}\left(1-e_1e_2-e_2e_3+e_3e_1\right)$, \\
&$\frac{1}{2}\left(1+e_1e_2-e_2e_3-e_3e_1\right)$, 
$\frac{1}{2}\left(1+e_1e_2+e_2e_3+e_3e_1\right)$,\\
&$\frac{1}{2}\left(1+e_1e_2-e_2e_3+e_3e_1\right)$, 
$\frac{1}{2}\left(1+e_1e_2+e_2e_3-e_3e_1\right)$, \\
&$\frac{1}{2}\left(1-e_1e_2+e_2e_3+e_3e_1\right)$, 
$\frac{1}{2}\left(1-e_1e_2-e_2e_3-e_3e_1\right)$\\
\hline\noalign{\smallskip}
$\mathbf{8}_-$&
$-\frac{1}{2}\left(1-e_1e_2+e_2e_3-e_3e_1\right)$, 
$-\frac{1}{2}\left(1-e_1e_2-e_2e_3+e_3e_1\right)$, \\
&$-\frac{1}{2}\left(1+e_1e_2-e_2e_3-e_3e_1\right)$, 
$-\frac{1}{2}\left(1+e_1e_2+e_2e_3+e_3e_1\right)$,\\
&$-\frac{1}{2}\left(1+e_1e_2-e_2e_3+e_3e_1\right)$, 
$-\frac{1}{2}\left(1+e_1e_2+e_2e_3-e_3e_1\right)$, \\
&$-\frac{1}{2}\left(1-e_1e_2+e_2e_3+e_3e_1\right)$,
$-\frac{1}{2}\left(1-e_1e_2-e_2e_3-e_3e_1\right)$\\
\hline\noalign{\smallskip}
$\mathbf{6}$&$\pm e_1e_2$, $\pm e_2e_3$, $\pm e_3e_1$\\
\hline\noalign{\smallskip}
$\mathbf{12}$&$\frac{1}{\sqrt{2}}(\pm e_1\pm e_2)$, $\frac{1}{\sqrt{2}}(\pm e_2\pm e_3)$, $\frac{1}{\sqrt{2}}(\pm e_3 \pm e_1)$\\
\hline\noalign{\smallskip}
$\mathbf{6}_+$&$\frac{1}{\sqrt{2}}(I\pm e_1)$, $\frac{1}{\sqrt{2}}(I\pm e_2)$, $\frac{1}{\sqrt{2}}(I\pm e_3)$\\
\hline\noalign{\smallskip}
$\mathbf{6}_-$&$-\frac{1}{\sqrt{2}}(I\pm e_1)$, $-\frac{1}{\sqrt{2}}(I\pm e_2)$, $-\frac{1}{\sqrt{2}}(I\pm e_3)$\\
\hline\noalign{\smallskip}
\end{tabular}
\end{table}

In this section we return to the versor description of the tetrahedral groups, and turn from the spinor to the pinor case, i.e. including reflections. 
Free multiplication of the root vectors results in a pinor group $\Pin(A_3)$ of order 48, which is a double cover of the Coxeter group $A_3$ of reflection symmetries of the tetrahedron of order 24. The 8 conjugacy classes of $\Pin(A_3)$ are listed in Table \ref{tab:2}. 

The elements $1$ and $-1$ are again in conjugacy classes by themselves ($\mathbf{1}$ and $\mathbf{1}_-$), and the $\Spin(A_3)=2T$  class $\mathbf{6}$ of bivectors is unaffected and remains a class in  $\Pin(A_3)$.   There are two conjugacy classes of order 8, $\mathbf{8}_+$ and $\mathbf{8}_-$, that contain the merged $2T$ classes $\mathbf{4}$ and $\mathbf{4}^{-1}$, and $\mathbf{4}_-$ merged with $\mathbf{4}^{-1}_-$. The new pinor elements are a set of 12 pure vectors, which are precisely the 12 root vectors, which form a conjugacy class by themselves  $\mathbf{12}$. 
The remaining 12 elements come as a pair of $\mathbf{6}_+$ and $\mathbf{6}_-$ that contain vector and trivector components and that are the negatives of one another. 
As we said earlier, since $A_3$ does not contain the inversion, the pin group does not simply consist of two copies of the spin group $\Spin(A_3)$ (e.g. unlike for $B_3$ and $H_3$). 
The resulting  orders are hence $1$, $1$, $6$, $6$, $6$, $8$, $8$ and $12$. Since these pinors doubly cover the full tetrahedral group $A_3$ of order $24$ via sandwiching, the conjugacy classes of $\Pin(A_3)$ collapse down to those of $A_3$ as follows: $1$ and $-1$ are the same under double-sided action and give the identity in $A_3$. The six bivectors in $2T$ are again halved, i.e. a $\mathbf{3}$ in $A_3$.  The two conjugacy classes of order 8 give the same rotations, so result in an $\mathbf{8}$ in $A_3$; likewise the remaining two $\mathbf{6}_+$ and $\mathbf{6}_-$ give a $\mathbf{6}_\pm$ in $A_3$, whilst the 12 root vectors $\mathbf{12}$ only give 6 different reflections, i.e. another  $\mathbf{6}$. This gives the expected five $A_3$ conjugacy classes  $\mathbf{1}$, $\mathbf{3}$, $\mathbf{6}$, $\mathbf{6}_\pm$ and $\mathbf{8}$.

\section{The conformal group and the modular group}\label{sec_CGA}

The versor formalism is particularly powerful in the CGA approach \cite{HestenesSobczyk1984, LasenbyDoran2003GeometricAlgebra,Dechant2011Thesis}, as translations and other conformal transformations can also be handled multiplicatively as versors, which is interesting for lattices and quasicrystals \cite{Hestenes2002PointGroups,Hestenes2002CrystGroups, Hitzer2010CLUCalc, Dechant2012AGACSE, DechantTwarockBoehm2011E8A4, DechantTwarockBoehm2011H3aff}. The conformal space of signature $(+,+,+,-)$ is given by adding two orthogonal unit vectors $e$ and $\bar{e}$ to the algebra of the plane \cite{HestenesSobczyk1984}. It is therefore spanned by the unit vectors 
$    e_1, e_2,  e, \bar{e}, \text{ with } e_i^2=1, e^2=1, \bar{e}^2=-1,  n := e +\bar{e}, \,\,\,  \bar{n}:= e -\bar{e}, \,\,\, N:= e\bar{e}.
$ The 2D vector $x=x_1e_1+x_2e_2$ is then represented as a  null vector in the conformal space by defining 
 $ X\equiv F(x):= x^2 n+2 x-\bar{n}$ 
, and we enforce the normalisation $X\cdot n = -1$.
A null $X$ allows for a homogeneous (projective) representation of points. 
The 2D conformal group is 6-dimensional and consists of 2 translations $x\rightarrow x+a$ by a vector $a$ which are given by a spinor 
$T_a=\exp\left(\frac{na}{2}\right)=1+\frac{na}{2}$ via $T_a F(x) \tilde{T}_a=F(x+a)$, the usual rotation in the plane $R F(x) \tilde{R}=F(Rx\tilde{R})$, a dilation generated by $D_\alpha=\exp\left(\frac{\alpha N}{2}\right)=\cosh \frac{\alpha}{2}+\sinh\frac{\alpha}{2}N$ and two special conformal transformations $K_a=eT_ae=1-\frac{\bar{n}a}{2}$, where we can just take $a$ to be $e_1$ and $e_2$ for simplicity; reflections in a (spatial) vector $a$ are not needed in the connected component of the conformal group as such, but are of course the fundamental building block given by $-aF(x)a$, and inversions are given by reflecting in $e$, $-eF(x)e$. 

Clifford analysis does not appear to  have been applied to the particularly relevant area of 2D conformal field theory yet, which is the most important in the context of string theory, and has instead focused on general features in arbitrary dimensions rather than low-dimensional accidental properties,  such as the fact that the 2D conformal algebra is infinite which makes (anti)holomorphic functions a central concept in 2D CFT. For the usual definitions of $z=x+Iy$ and $\frac{\partial}{\partial z}=\frac{1}{2}\left(\frac{\partial}{\partial x}-I\frac{\partial}{\partial y}\right)$ in Clifford analysis and defining the vector fields $L_n=-z^{n+1}\frac{\partial}{\partial z}$ one immediately recovers the Witt algebra $\left[L_m, L_n\right]=(m-n)L_{m+n}$. In string theory, its central extension the Virasoro algebra $\left[L_m, L_n\right]=(m-n)L_{m+n}+\frac{c}{12}\delta_{m+n,0}m(m^2-1)$ is of central importance. 

In the context of strings propagating on a torus, or for the string one-loop scattering amplitude, modularity is important. The modular group is a subgroup of the conformal group. Given the complex structure of a torus $\tau$  one can embed the torus as a parallelogram tiling of the complex plane, with periodicity conditions across opposite edges, resulting in modularity in the complex plane: the torus does not change under wrapping around a circular direction, corresponding to a symmetry generator $T:\tau\rightarrow \tau+1$, nor under swapping the two circular directions, with the generator $S:\tau\rightarrow -1/\tau$. $S$ and $T$ generate the modular group; its  defining relations are $\langle S,T | S^2=I, (ST)^3=I\rangle$, i.e. it is the triangle group $(2,3,\infty)$.
 We are particularly interested in Moonshine phenomena, which are  correspondences between the two very different areas of finite simple groups and modular forms; e.g. monstrous moonshine  \cite{gannon2006moonshine}. The dimensions of the smallest irreducible representations of the Monster group $M$ are $1, 196883, 21296876, \dots$. Since the torus is periodic in two directions, one can Fourier expand and get a set of coefficients for any modular function; in particular for the generating function of modular functions, the Klein {$j(\tau)$} modular function, one gets $j(\tau)=q^{-1}+744+196884q+21493760q^2+\dots$, where $q=e^{2\pi i\tau}$. Monstrous moonshine is then the observation that these two seemingly very different areas of mathematics are intimately related via $196884=196883+1$, $21493760=21296876+196883+1$ and so on. The Moonshine phenomenon we are currently interested in is somewhat different \cite{eguchi2011notes,taormina2013twist}: it concerns as the analogue of the modular form (in fact a weak Jacobi form) the {elliptic genus} of an $\mathcal{N}=4$ superconformal field theory compactified on a {$K3$-surface}; the elliptic genus of $K3$ has been known since the 80s and looks very unsurprising, unless one rewrites it in an inspired way in terms of the $\mathcal{N}=4$ characters \cite{Petersen1991modular,Petersen1990characters2,Eguchi1987unitary,Eguchi1988character,Ooguri1992modular} as ${E_{K3}(\tau,z)=-2Ch(0;\tau,z)+20Ch(1/2;\tau,z)+e(q)Ch(\tau, z)}$ where
all the coefficients in the $q$-series ${e(q)=90q+462q^2+1540q^3+4554q^4+11592q^5+\dots}$ are twice the dimension of some irreducible representation of the Matthieu group $M_{24}$ whose   irreps have dimensions $45, 231, 770, 2277, 5796 \dots $). Modularity is therefore very topical, also in other areas \cite{wiles1995modular}, and a Clifford perspective on holomorphic and modular functions could have rather profound consequences. 

In CGA, the vector $x=x_1e_1+x_2e_2$ is related to the complex structure $\tau=x_1+x_2e_1e_2$ via $\tau=e_1x$.
The $T$ generator acting on $x$ is simply given by $Tx=(x_1+1)e_1+x_2e_2$, which is achieved via the conformal spinor $T_X=1+\frac{ne_1}{2}$. $S$ acts as $x\rightarrow\frac{1}{x_1^2+x_2^2}(-x_1e_1+x_2e_2)$. This corresponds to an inversion $x\rightarrow\frac{x}{x^2}$, followed by a reflection in $e_1$, $-e_1F(x)e_1$. Thus, $S$ is given by $S_X=e_1e$ via $S_XF(x)\tilde{S}_X$. Since $e_1ee_1e=-1$ it is  obvious that $S$ satisfies the first defining relation $S^2=I$. It is also easy to show that $(S_XT_X)^3=-1$ in CGA, so the second defining relation is also satisfied, and $ST$ is essentially a 3-fold rotation in conformal space (see Section \ref{sec_bin}). We will explore applications of this conformal representation of the modular group for modular forms in future work.

\section{Conclusions}\label{sec_concl}

In this paper, we have discussed applications of a Clifford algebra framework to groups connected with reflections. We described a versor framework for these, and explored the Coxeter group $A_3$ and the related tetrahedral groups in detail. We presented the induction theorem for 4D root systems and discussed their automorphism groups. We finished with a discussion of reflections in conformal space and extending these to the conformal and modular groups. These could have interesting applications for holomorphic and modular functions in particular for conformal field theory and string theory on the torus, which we will discuss  in further work. 


\subsection*{Acknowledgment}
I would like to thank Reidun Twarock, Anne Taormina, David Hestenes, Anthony Lasenby, Joan Lasenby, Eckhard Hitzer, David Eelbode and Vladimir Soucek. 


\end{document}